\documentclass[11pt]{article}

\usepackage[utf8]{inputenc}

\usepackage[margin=2.5cm]{geometry}

\usepackage{amsmath,amssymb,amsthm}
\usepackage{mathtools}

\usepackage[linesnumbered,ruled,vlined]{algorithm2e}

\usepackage[svgnames,table]{xcolor}
\usepackage{graphicx}
\usepackage{tikz}
\usepackage{tikz-qtree}
\usetikzlibrary{arrows.meta}
\usepackage{pgfplots}
\pgfplotsset{compat=1.18}

\usepackage{forest}
\usepackage{tcolorbox}
\usepackage{subcaption}
\usetikzlibrary{decorations.pathreplacing,calc}

\tcbset{
    boxrule=0.8pt
}

\forestset{
    solid nodes/.style={
        for tree={circle,draw,inner sep=1,fill=green}
    },
    dir/.style={
        for tree={grow=#1}
    },
    leaf/.style={
        label=$#1$
    },
    mytree/.style={
        solid nodes,
        for tree={grow=south,s sep=1cm}
    }
}

\usepackage{array}
\usepackage{booktabs}
\usepackage{float}
\usepackage{hhline}
\usepackage{multirow}

\usepackage[round]{natbib}

\usepackage{nicefrac}
\usepackage{textcomp}
\usepackage{enumerate}
\usepackage{microtype}
\usepackage{tabto}

\usepackage[colorlinks=true,linkcolor=DarkBlue,citecolor=DarkBlue]{hyperref}

\usepackage{appendix}

\newtheorem{theorem}{Theorem}[section]
\newtheorem{claim}[theorem]{Claim}

\newtheorem{lemma}[theorem]{Lemma}

\theoremstyle{definition}
\newtheorem{definition}[theorem]{Definition}

\theoremstyle{remark}

\newcommand{\calN}{\mathcal{N}}

\definecolor{Blue}{HTML}{2D2F92}

\newcommand*{\eps}{\ensuremath{\epsilon}}

\newcommand{\notelist}{}

\NewDocumentCommand{\addnote}{mm}{
  \expandafter\gdef\expandafter\notelist\expandafter{%
    \notelist
    \noindent \hyperlink{#2}{#1}\par
  }
  \hypertarget{#2}{#1}%
}

\usepackage{amsmath,amsfonts,bm}

\def\eqref#1{equation~\ref{#1}}
\def\1{\bm{1}}

\def\eps{{\epsilon}}

\DeclareMathAlphabet{\mathsfit}{\encodingdefault}{\sfdefault}{m}{sl}
\SetMathAlphabet{\mathsfit}{bold}{\encodingdefault}{\sfdefault}{bx}{n}

\newcommand{\R}{\mathbb{R}}

\newtheorem{problem}[theorem]{Problem}

\newcommand{\eat}[1]{}

\newcommand{\Brac}[1]{\left[#1\right]}

 \newcommand{\set}[1]{\left\{ #1 \right\} }

\newcommand{\norm}[1]{\lVert #1 \rVert}

\newcommand{\bignorm}[1]{\big\lVert#1\big\rVert}

\newcommand{\bigsnormt}[1]{\bignorm{#1}^2_2}

\newcommand{\Iprod}[1]{\left\langle#1\right\rangle}

\newcommand{\Esymb}{\mathbb{E}}
\newcommand{\Psymb}{\mathbb{P}}

 \DeclareMathOperator*{\ProbOp}{\Psymb}

\renewcommand{\Pr}[1]{\ProbOp\left(#1\right)}

\renewcommand{\epsilon}{\varepsilon}

\newif\ifnotes\notestrue
\ifnotes
\usepackage{color}
\definecolor{mygrey}{gray}{0.50}
\newcommand{\notename}[2]{{\textcolor{blue}{\footnotesize{\bf (#1:} {#2}{\bf ) }}}}

\else

\newcommand{\notename}[2]{{}}

\fi

\newcommand{\VC}{\mathrm{VCdim}}

\title{Optimal VC Dimension of Contrastive Learning with Margin}

\author{Dionysis Arvanitakis\\Northwestern University \and Vaggos Chatziafratis\\UC Santa Cruz \and  Yiyuan Luo\\UC Santa Cruz\and Konstantin Makarychev\\Northwestern University}
\date{}

\begin{document}
\maketitle
\begin{abstract}
Contrastive learning is a successful paradigm for learning $d$-dimensional geometric representations from a collection of ``anchor--positive--negative'' triplets $(i,j^{+},k^{-})$, indicating that “item $i$ is closer to $j$ than to $k$.” Despite its success, understanding why contrastive learning leads to representations of high \textit{generalization} quality---beyond the often pessimistic predictions from PAC-learning---remains a central question. Recently, \citet*{alon2024optimal} proved that, for PAC-learning $d$-dimensional Euclidean representations of $n$-point datasets, $\Theta(\min(nd, n^2))$  triplets are necessary and sufficient, while they posed as an open question whether  their VC dimension bounds for the more realistic setting of \textit{contrastive learning with a margin} can be improved. For a margin parameter $\alpha >0$, a triplet $(i,j^{+},k^{-})_{\alpha}$ is satisfied by the embedding $\phi:[n]\rightarrow \mathbb{R}^{d}$, if $\|\phi(i)-\phi(k)\|_2>(1+\alpha)\cdot\|\phi(i)-\phi(j)\|_2$. In this work, we resolve their question by proving that the VC dimension of contrastive learning under any margin $\alpha\in(0,1)$ is in fact $O(n/\alpha^2)$, improving on the previous bound of $O(n\log(n)/\alpha^2)$. We also establish that the bounds are optimal up to constant factors, by providing a matching lower bound of $\Omega(\frac{n}{\alpha^2})$ (the previously known lower bound was $\Omega(\frac{n}{\alpha})$), for $\alpha\geq \max(n^{-1/2},d^{-1/2})$.
\end{abstract}

\newpage
\tableofcontents
\newpage
\thispagestyle{empty}
\section{Introduction}
%should we maybe say contrastive embeddings?
Contrastive learning is a powerful paradigm in representation learning, where the goal is to find meaningful representations, encoded as features in $d$-dimensional Euclidean space, allowing us to replace data point $x$ by its feature vector $\phi(x)\in \mathbb{R}^d$ in various downstream tasks. At the heart of contrastive learning lies the notion of a \textit{triplet} comparison $(i,j^+,k^-)$ indicating that “item $i$ is closer to $j$ than to $k$,” also known as ``anchor--positive--negative''~\citep{schroff2015facenet}. For example, modern pipelines are trained with a contrastive loss, e.g., Triplet or InfoNCE (with one or more negatives), mapping data into $d$-dimensional vectors by pulling anchor-positive pairs closer, while pushing anchor-negative pairs further apart, enabling a unified approach for diverse types of data, including text, audio, graph-based tasks, and images~\citep{schroff2015facenet,grover2016node2vec,oord2018representation,saunshi2019theoretical,chen2020simple}. 

Unfortunately, despite its empirical success, many theoretical aspects of contrastive learning are not well-understood. In this work, we focus on \textit{sample complexity} and \textit{generalization bounds}---a central question in PAC-learning---making progress towards understanding when a representation learned from anchor--positive--negative samples generalizes well to unseen (test) samples. 

%\dnote{Do we want to make the transition from previous paragraph a bit smoother? Right now it feels a bit weird that this formalization is not exactly capturing what is described in the previous paragraph.}

\subsection{Sample Complexity of Contrastive Learning} 

Following the framework by~\cite*{alon2024optimal}, let $\mathcal X$ be an $n$-point dataset and let $\phi:\mathcal{X}\to \mathbb{R}^{d}$ be an embedding that maps points into $d$-dimensional Euclidean space (we do not make any assumptions on $\phi$). Motivated by the anchor--positive--negative paradigm, typically we assume access to a collection of $m$ triplets; these are usually generated from data, such as nearby words in a document, or neighboring video frames, that serve as the anchor-positive pair, and then a random other data point plays the role of the negative example. For more on how to pick negative samples, and their role in downstream performance, please refer to~\cite{saunshi2019theoretical,robinson2020contrastive,awasthi2022more}.

\begin{definition}
  We say that the embedding $\phi$ \emph{satisfies} or \emph{is consistent with} input triplet $(i,j^{+},k^{-})$ if:% and a training set of $m$ triplets $S \subseteq [n]^3$, we seek an embedding $\phi:\mathcal X\rightarrow \mathbb{R}^{d}$ so that $i$ is closer to $j$ than $k$:
\[
||\phi(i) - \phi(k)||_2 > ||\phi(i) - \phi(j)||_2.
\]
%In this case, the embedding $\phi$ is said to \textit{satisfy} the input triplet $(i,j^+,k^-)$.
More generally, it is common to enforce a small \textit{margin} in the above distance comparison ~\citep{vankadara2023insights}, i.e., for a margin parameter $\alpha>0$, the embedding $\phi$ satisfies the triplet $(i,j^{+},k^-)_{\alpha}$ if:
\[
||\phi(i) - \phi(k)||_2 > (1+\alpha) ||\phi(i) - \phi(j)||_2.
\]  
\end{definition}

The basic question motivating our work concerns the sample complexity of contrastive learning: \textit{for a set of $n$ items, what is the smallest number $m^*(\varepsilon, \delta)$ of anchor--positive--negative samples required to PAC-learn an underlying ground-truth embedding $\phi(x)\in \mathbb{R}^d$?} 

In particular, we focus on the following problem in PAC-learning studied by~\citet*{alon2024optimal}: 

\begin{problem}[PAC-Learning an embedding $\phi(\cdot)$ from triplets~\cite{alon2024optimal}]\label{prob:PAC} Fix error parameters $\epsilon,\delta>0$.
Given i.i.d. samples $(i,j,k)$ drawn from a distribution $\mathcal{D}$ over $[n]^{3}$, and labeled according to embedding $\phi$ as either $(i,j^{+},k^-)$ or $(i,k^{+},j^{-})$, based on whether $||\phi(i) - \phi(j)||_2\lessgtr||\phi(i) - \phi(k)||_2$, the goal is to find an embedding $\hat{\phi}$\footnote{In the contrastive learning with margin setting studied below, learning might require an improper predictor.} that generalizes well to unseen triplets, i.e., with probability $1-\delta$, for a subsequent  sample $(i,j,k)$ sampled from $\mathcal{D}$:
\[
\mathbb{P}_{(i,j,k)\sim \mathcal D}[\text{ $\hat{\phi}$ disagrees with $\phi$ on $(i,j,k)$}]\le \epsilon.
\]
\end{problem}
In other words, for a test triplet query $(i,j,k)$, the learned embedding $\hat \phi$ must classify whether the anchor $i$ is closer to $j$ or to $k$ according to the (unknown) $\phi$, i.e., whether the triplet is $(i,j^+,k^-)$ or  $(i,k^+,j^-)$. Similarly, Problem~\ref{prob:PAC} can be studied in the agnostic case (where samples are not necessarily consistent with some embedding), or with quadruplet samples $(i,j,k,l)$ indicating whether $||\phi(i) - \phi(j)||_2\lessgtr||\phi(k) - \phi(l)||_2$ from the literature of ordinal embeddings~\citep{bilu2005monotone,alon2008ordinal,terada2014local,vankadara2023insights}.% As we will see, our results are obtained via novel VC-dimension bounds, and extend to these cases, improving bounds by~\cite{alon2024optimal}.  %In other words, we are interested in the standard PAC-learning generalization setting that bounds the probability of error on an (unseen) test triplet query $(i,j,k)$, sampled from the same distribution as $m$ input triplets (training set), where 
%maybe say somtething about phi hat? need to be careful that phi and had phi have different domains, one points the other triplets

%-generalization to unseen triplets so it is useful as it bound the probability of error and relevant for practice and theory

 %as mentioned above, the label of a triplet $(i,j,k)$ is either $(i,j^+,k^-)$ or $(i,k^+,j^-)$. %(see Section~\ref{sec:overview} for formal definitions of the concept classes).
%\vspace{-0.2cm}

\paragraph{Dimensionality-vs-Generalization.} Recently, Alon, Avdiukhin, Elboim, Fischer and Yaroslavtsev~\citep{alon2024optimal} showed that the VC-dimension of the learning task in Problem~\ref{prob:PAC} is $O(\min(nd, n^2))$, i.e., they proved that $m^*(\varepsilon, \delta)=O_{\varepsilon,\delta}(\min(nd, n^2))$ labeled triplets are sufficient\footnote{The  bound is $O(nd)$, since $d\le n$ without affecting pairwise distance comparisons~\citep{bilu2005monotone}.} for PAC-learning the embedding $\phi: \mathcal{X}\rightarrow \mathbb{R}^{d}$. Typically, $\phi$ stems from the output layer of a neural network, so they focused on the representation dimensionality $d$, a key parameter of interest often ranging in the thousands~\cite{radford2018improving,kusupati2022matryoshka,comanici2025gemini}. 

%managed to overcome the overall lack of meaningful generalization bounds in the context of contrastive learning. They 

%Unfortunately, given that the embedding $\phi(\cdot)$ typically stems from the output layer of a neural network, it is well-known that direct application of PAC-learning generalization bounds often leads to very loose bounds, due to the high capacity of modern deep learning architectures \dnote{I think we should be a bit careful here, whatever bounds we prove here for general embeddings, only better bounds hold for embeddings induced by neural nets.}. 

%Our Contribution: Generalization beyond JL lemma from O(n) sample
\vspace{-0.2cm}
\subsection{Our Contribution: Generalization from \texorpdfstring{$O(n)$}{O(n)} samples, under \textit{any} margin}
For the special case when the margin $\alpha=0$ between distances $||\phi(i) - \phi(j)||_2, ||\phi(i) - \phi(k)||_2$, \cite{alon2024optimal} managed to show a matching lower bound of $\Omega(\min(nd, n^2))$ for Problem~\ref{prob:PAC}. However, it is unclear what happens for the well-motivated case of contrastive learning with margin $\alpha>0$. Observe that $\alpha>0$ provides some gap between distances, so the dimension $d$ can be directly replaced by $O(\log n/\alpha^2)$ using standard tools, e.g., Johnson-Lindenstrauss (JL) lemma: as noted in~\citep{alon2024optimal}, JL yields $d=O(\log n/\epsilon^2)$, where $\epsilon$ is the distortion~\citep{dasgupta2003elementary}, so $O(\tfrac{n\log n}{\alpha^2})$ triplets suffice for generalization (set distortion $\epsilon\approx \alpha$). Improving this bound was raised as an open question by~\cite{alon2024optimal}:

\begin{center}
   \emph{For contrastive learning with large enough values of margin $\alpha>0$, \\ can the VC-dimension bound be improved beyond  $O(\tfrac{n\log n}{\alpha^2})$?} 
\end{center}

%\vnote{editing below, feel free to read/add comment up to here}

%While for the case of no margin ($\alpha=0$) this approach may be detrimental (evident by the exploding term as $\alpha\to0$), potentially flipping almost all triplet distance comparisons due to the inevitable $(1\pm\epsilon)$-distortion~\cite{bilu2005monotone,alon2008ordinal}, it is still unclear whether better 

\paragraph{Our Main Result.} We go beyond the above approach based on JL, and we obtain optimal $O(n)$ bounds in the VC dimension, thus resolving the above question for \textit{all} values of margin $\alpha>0$ (not necessarily large values as originally stated in~\cite{alon2024optimal}). 

Our main result is a stronger, linear in $n$ VC-bound of $O(\max\set{n/\alpha^2,n})$ for contrastive learning. This is coupled with a matching lower bound of $\Omega(\frac{n}{\alpha^2})$ (for $\alpha=\Omega(\max\set{n^{-1/2},d^{-1/2}})$), so when put together the two results give a tight characterization of the VC dimension in this setting. To make things formal, we need the following definition related to the notion of \textit{partial} concept classes from~\cite{alon2022theory}:

%\begin{definition}[Concept class $\mathcal{H}(n,d,\alpha)$ for embeddings with margin] Let $\mathcal X$ be an $n$-point dataset, let $\phi:[n]\to \mathbb{R}^{d}$ be an embedding, and let $\alpha>0$ be a margin parameter. Associated with embedding $\phi(\cdot)$ is its triplet classifier $h_{\phi}: [n]^{3} \to \set{+1,-1,\star}$, labeling\footnote{The symbol $\star$ is arbitrary, denoting that classifiers may not answer certain queries (e.g., if distances are not separated).} an input triplet as follows:
%    \[h_{\phi}(i,j,k) = \begin{cases} +1, & \text{if } \|\phi(i)-\phi(k)\|_2 > (1+\alpha)\|\phi(i)-\phi(j)\|_2, \\[6pt] -1, & \text{if } \|\phi(i)-\phi(j)\|_2 > (1+\alpha)\|\phi(i)-\phi(k)\|_2, \\[6pt] \phantom{-}\star, & \text{otherwise.} \end{cases}\] 
%    The concept class that contains all possible triplet classifiers is denoted as $\mathcal{H}(n,d,\alpha)$. 
%\end{definition}

\begin{definition}[Concept class $\mathcal{H}(n,d,\alpha)$ for embeddings with margin]\label{def: partial} Fix $\alpha>0$, and let $\phi:[n]\to \mathbb{R}^{d}$ be an embedding. Associated with embedding $\phi(\cdot)$ is its triplet classifier $h_{\phi}: [n]^{3} \to \set{+1,-1,\star}$, labeling\footnote{Symbol $\star$ is arbitrary, meaning classifiers may not answer certain queries (e.g., if distances are not separated).} an input triplet as follows:
    \[h_{\phi}(i,j,k) = \begin{cases} +1, & \text{if } \|\phi(i)-\phi(k)\|_2 > (1+\alpha)\|\phi(i)-\phi(j)\|_2, \\[6pt] -1, & \text{if } \|\phi(i)-\phi(j)\|_2 > (1+\alpha)\|\phi(i)-\phi(k)\|_2, \\[6pt] \phantom{-}\star, & \text{otherwise.} \end{cases}\] 
    The concept class containing all possible triplet classifiers is denoted\footnote{For the case of no margin, we simply write $\mathcal{H}(n,d)$. Note that for embeddings $\phi$, as in the definition, $\mathcal{H}(n,d)$ is a total concept class.} as $\mathcal{H}(n,d,\alpha)$. We only consider embeddings such that for every $i,j,k$ with $j\ne k$, $\|\phi(i)-\phi(j)\|\neq \|\phi(i)-\phi(k)\|$.
\end{definition}
\begin{theorem}[Linear VC-dimension for \textit{any} margin]\label{th:intro}
%Let $\alpha\in (0, 1].$ 
The VC-dimension of $\mathcal{H}(n,d,\alpha)$ is $O(\frac{n}{\alpha^2})$ for $\alpha\in(0,1)$ and $O(n)$ for $\alpha\geq 1$. Specifically, there exists absolute constant $C$ such that:
    \begin{align*} 
        \mathrm{VCdim}(\mathcal{H}(n,d,\alpha))\leq  C\cdot\max\bigl\{\frac{n}{\alpha^2},n\bigr\}.
    \end{align*}

Furthermore, assuming that $\alpha\geq\max(\frac{1}{\sqrt{n}}, \frac{1}{\sqrt{d}})$, the VC dimension is $\Omega(\frac{n}{\alpha^2})$, for $\alpha\leq 1$ and $\Omega(n)$ for $\alpha>1$. In particular there exists an absolute constant $c$ such that:
\begin{align*}
      \mathrm{VCdim}(\mathcal{H}(n,d,\alpha))\geq c\cdot \max\{\frac{n}{\alpha^2}, n\}.
\end{align*}
\end{theorem}
We prove the upper bound of Theorem \ref{th:intro} in Section \ref{sec: proof of main theorem}, the proof of the lower bound follows the construction for linear classification with margin and is proven in Section \ref{sec: app lower bound proof}.

A few remarks are in order. First, our result  characterizes the VC-dimension of concept class $\mathcal{H}(n,d,\alpha)$ in the regime where $\alpha\geq \max(n^{-1/2}, d^{-1/2})$ as $\mathrm{VCdim}(\mathcal{H}(n,d,\alpha))=\Theta( \max(n,\frac{n}{\alpha^2}))$. Second, it is surprising that under a margin $\alpha$, dimensionality does not play a role: the larger the margin, the easier the learning task becomes, since the sought-after embedding $\phi$ is less refined, and a priori, one might expect that dimension $d$ plays a role on the samples required to avoid misclassification of test samples (perhaps this is why~\cite{alon2024optimal} originally asked the question for large margins as a simplification). 

\paragraph{Learning Bounds via Geometric Approximation.} An interesting aspect in our analysis that allows us to obtain the $O(n)$ VC-bound, is the use of geometric approximation algorithms aiming to satisfy as many of the input triplets as possible. In contrast to the approach taken in~\cite{alon2024optimal}, where they used combinatorial bounds on the number of sign-patterns of a system of polynomials, here we essentially reduce the $d$-dimensional learning problem to a $1$-dimensional setting, whose sample complexity we can bound.
On the one hand, we exploit the margin $\alpha>0$ to show that a simple one-dimensional projection satisfies $\tfrac12 +\Omega(\alpha)$ fraction of triplets, independently of dimension $d$, where $\tfrac12$-approximation is a trivial baseline achieved by a random embedding.
On the other hand, we use the probabilistic method to show that no set of size $m=\Omega(\frac{n}{\varepsilon^2})$ can be shattered on the line $d=1$, and more precisely, there exists a labeling of the triplet samples such that no embedding satisfies more than $\frac{1}{2}+\varepsilon$ of the triplets; since $\varepsilon$ can be taken to be sufficiently smaller than $\alpha$ this leads to a contradiction, giving us the upper bound of Theorem~\ref{th:intro}. Our approach is similar in spirit to the algorithm of \cite{blum2005random} (see also \cite{balcan2005pac}) for linear classification with margin which combines a random projection coupled with boosting. For formal definitions and a detailed technical overview see Section \ref{sec:overview} below.

\subsection{Further Related Work}
There are intense efforts to understand various aspects of contrastive learning, both from a theoretical and a practical perspective. For example, many well-known works study relations between transfer learning, inductive biases and contrastive learning~\citep{saunshi2019theoretical,saunshi2022understanding}, posing specific assumptions on latent classes or how contrastive triplets are generated, in order to try to explain the observed success of such methods~\citep{schroff2015facenet,oord2018representation,chen2020simple}.\begin{samepage} Other works focus more on optimization loss functions such as spectral contrastive loss with good downstream performance~\citep{haochen2021provable}, understanding the role of negative samples~\citep{robinson2020contrastive,awasthi2022more,lei2023generalization}\end{samepage}, adversarial training~\citep{zou2023generalization} and more. Here, the work closest to ours is~\cite{alon2024optimal}, which served as the main inspiration towards our sample complexity results. 

Moreover, the literature of non-metric embeddings, often called \textit{ordinal} embeddings or \textit{monotone} maps~\citep{bilu2005monotone,alon2008ordinal,terada2014local,emamjomeh2018adaptive,ghosh2019landmark,ghoshdastidar2019foundations,fan2020learning,avdiukhin2023tree,chatziafratis2024dimension,avdiukhin2024embedding}, is also closely related to our setting, as their main algorithmic task is, given a collection of comparisons between pairwise distances, to find an embedding into a suitable geometric space, e.g., tree metric,  ultrametric, or Euclidean space, so that as many of the given pairwise comparisons are satisfied. We note that, because humans are much faster at comparisons between objects, in contrast to providing numerical scores, such triplet comparisons are also widely-deployed in metric learning and crowdsourcing~\citep{schultz2003learning,tamuz2011adaptively,kulis2013metric}. Other types of comparisons include the quadruplet setting $(i,j,k,l)$ mentioned earlier, where comparisons are of the form ``item $i$ is closer to $j$, than item $k$ is to $l$,'' see survey~\cite{vankadara2023insights}. 

%Complexity of Finding Local Optima in Contrastive Learning

Finally, there are several results on the computational complexity and approximation algorithms literature associated with the empirical risk minimizer on a given sample. Given $m$ (non-contradictory) triplet constraints on $n$ items,~\cite{bilu2005monotone} show that $n-1$ dimensions always suffice to satisfy them, while $\Omega(n)$ are necessary. In~\cite{avdiukhin2024embedding}, they present an algorithm that achieves dimension $\sqrt{m}$ and further improve it via the notion of graph arboricity. For the case of non-realizable instances, recent work by~\cite{arvanitakis2026provableaccuracycollapseembeddingbased} showed that it is computational intractable to obtain a better than $\tfrac12$-approximation (assuming the Unique Games Conjecture~\citep{khot2002power}), showing that the problem of maximizing satisfied triplets is approximation resistant~\citep{haastad2001some,hast2005beating,guruswami2008beating}. 

Regarding more practical approaches based on local search, including gradient-based methods, it is known that in worst-case, finding even a locally optimum solution is intractable, as captured by PLS-completeness\footnote{Both PLS and CLS are complexity classes related to finding local optima via local search methods, see~\cite*{johnson1988easy, daskalakis2011continuous}.} for discrete settings, and CLS-completeness for continuous embedding settings~\citep{yan2026complexity}. However, it is reported in~\cite{vankadara2023insights}, that popular methods usually obtain near-zero training error for a variety of triplet loss objectives (see also~\citep{bower2018landscape} on spurious local optima vs global optima).

%also known as ordinal embeddings~\cite{bilu2005monotone,alon2008ordinal,vankadara2023insights} -contrastive embeddings, dimensionality tradeoffs, tree learning bounds from grigory, -representation learning generally, contrastive learning success -importance of dimension, yet in practice we get good generalization obviously. can we explain? we show under any margin, we get sample complexity independent of dimension -basic setting with triplets, pulling similar together, or apart - Luxburg survey and ordinal embeddings generally, and why JL fails - hard negatives paper and awasthi negatives paper %- is the bound tight for ell_1? -basic question: sample complexity, generalization: state main motivating question: can we explain success of representation contrastive learning? find paper that indicate regardless of the dimension they generalize well -state prior work Alon the main theorem and bottleneck with dimension -Contrastive Learning with a margin:  explain JL gives dimension logn dimension and how we will bypass beyond JL -contributions state main theorem, extensions to other related settings covering Alon et al., -describe a bit technical idea and consequences
\section{Technical 
Overview}\label{sec:overview}

We present the main proof strategy and technical ideas behind our improved $O(n/\alpha^2)$ bound on the VC-dimension for the concept class $\mathcal{H}(n,d,\alpha)$ of triplet classifiers with margin.
\subsection{Our approach for \texorpdfstring{$\alpha>0$}{α>0}: Algorithmic Advantage over Random}

To give some intuition and set notation, we first describe the approach taken in~\cite{alon2024optimal} for the case of zero margin, and why it fails in our case when margin $\alpha>0$.

\paragraph{Sign-Patterns of Polynomials vs. margin $\alpha>0$.} Given that each triplet $(i,j^+,k^-)$ imposes a \textit{geometric} constraint on how the representation $\phi$ embeds the $n$ data points, a first attempt would be to analyze how $\phi$ partitions $d$-dimensional Euclidean space, and count the regions wherein triplet classifiers do not switch their label. Indeed,~\cite{alon2024optimal} observed that each triplet is satisfied, if and only if the following polynomial $Q_{(i,j,k)}$ on $3\cdot d$ number of variables, and of degree $2$ is positive: 
\[
Q_{(i,j,k)}(\phi(i), \phi(j),\phi(k)) = \|\phi(i)-\phi(k)\|_2^2 - \|\phi(i)-\phi(j)\|_2^2
\] 
In total, this leads to $m$ such polynomials (as there are $m$ input triplets), with the total number of variables being $n\cdot d$. Then,~\cite{alon2024optimal} used known results from algebraic combinatorics to upper bound the different sign-patterns of the resulting system $\set{Q_{(i_t,j_t,k_t)}}_{t=1}^m$ of $m$ polynomials on $n\cdot d$ variables~\citep{warren1968lower,alon1985geometrical}. This limits the size of a set that can be shattered giving an upper bound on the VC-dimension (see definitions below).

While the above approach worked for the case of no margin $\alpha=0$, and yielded the bound $\Theta(nd)$, a similar approach for $\alpha>0$ would still yield the same (dimension-dependent) bound $\Theta(nd)$: notice that if instead we defined the polynomial
\[
Q'_{(i,j,k)}(\phi(i), \phi(j),\phi(k)) = \|\phi(i)-\phi(k)\|_2^2 - (1+\alpha)^2\|\phi(i)-\phi(j)\|_2^2,
\] 
the number of sign-patterns does not change, as those depend solely on $m$, the degree of the polynomials and the number of variables.

We take a different approach to exploit the margin through geometric approximation algorithms that estimate triplet labels induced by $\phi(i),\phi(j),\phi(k)\in \mathbb{R}^d$. On the one hand, the margin allows a $1$-dimensional random projection algorithm to get slightly better-than-random approximations, specifically satisfying $\approx\tfrac12+\Omega(\alpha)$ fraction of the triplets. This is crucially \textit{independent} of what the dimension $d$ is, and is the reason why our final bound in Theorem~\ref{th:intro} does not have a dependence on $d$. On the other hand, if the VC-dimension was substantially larger than $O(n)$, and the labels of the training set were chosen at random to be either $+1$ or $-1$, then no embedding could fit the labels of the input triplets, better than random chance, which leads to a contradiction.

\subsection{Rounding-vs-Noisy Labels in 1D}

To make the above precise, we remind the reader of the standard definitions of \textit{VC-dimension} and \textit{Shattering} of a set ~\citep{shalev2014understanding}, adapted to our setting. The appropriate framework for our results is that of ~\cite{alon2022theory}, where hypothesis functions can be \emph{undefined} on certain inputs (this corresponds to the asterisk case in Definition \ref{def: partial}).
\begin{definition}[Shattering]
We say that a set $S=\set{(i_t,j_t,k_t)}_{t=1}^{m}$ is shattered by $\mathcal{H}(n,d,\alpha)$ if for every $y: [m]\to \set{-1,1}$ there exists a concept function $h_{\phi}\in \mathcal{H}(n,d,\alpha)$ such that for every $t\in [m]$, $h_{\phi}(i_t,j_t,k_t)=y(t)$.
\end{definition}
Notice that in the definition, the concept is not allowed to output the special symbol $\star$ on any of the points. We can now give the definition of the VC dimension of the hypothesis class $\mathcal{H}(n,d,\alpha)$.
\begin{definition}[VC-dimension]
    The VC-dimension of hypothesis class $\mathcal{H}$ is the largest number $m$ such that there exists a set of size $m$ that is shattered. We will use $\mathrm{VCdim}(\mathcal{H})$ to denote the VC dimension of hypothesis class $\mathcal{H}$. 
\end{definition}

%\dnote{I don't know if we need this, I think it is explained earlier well}
%The main idea in the $O\left(\frac{n\log(n)}{\alpha^2}\right)$ bound of \cite{alon2024optimal} is that by the Johnson Lindenstrauss lemma, there is never a need to go beyond $O\left(\frac{\log(n)}{\alpha^2}\right)$: because of the separation, any embedding that satisfies a set of triplet constraints in $d$ dimensions can be projected to $O\left(\frac{\log(n)}{\alpha^2}\right)$ so that the distances are preserved well enough that all the constraints are satisfied. Replacing $d$ with $\approx \frac{\log(n)}{\alpha^2}$ in the upper bound $O(dn)$ for VC dimension of embeddings to $d$ dimensional space gives the result.

%\dnote{Say about shattering, defnition of VC dimension etc here??} \vnote{yes i think we can delete section 3 and more everything here essentially}
We directly apply the definition of the VC dimension and shattering to prove our main theorem. What we need to show is that for $m\approx \frac{n}{\alpha^2}$, no set of $m$ unlabeled triplets $\set{(i_t,j_t,k_t)}_{t=1}^{m}$ is shattered.
Equivalently, starting with $m$ arbitrary triplets $\set{(i_t,j_t,k_t)}_{t=1}^{m}$ (one can think of them as being selected by an adversary) we need to show that there is a labeling of the triplets such that no embedding can satisfy all of them.
We let $\varepsilon=\Theta(\alpha)$ be a sufficiently small constant fraction of $\alpha$ and let $m=\frac{n}{\varepsilon^2}=\Theta(\frac{n}{\alpha^2})$: we show that for a set $\set{(i_t,j_t,k_t)}_{t=1}^{m}$ of  unlabeled triplets there exists a labeling (i.e., every unlabeled triplet is labeled as either $(i_t,j^{+}_t,k_t^{-})_{\alpha}$ or $(i_t,k_t^{+},j_t^{-})_{\alpha}$) such that no embedding $\phi:[n]\to \mathbb{R}^{n}$ is consistent with all the labeled triplets.

\paragraph{Rounding to 1D.} Our proof, perhaps surprisingly, begins by an observation drawing ideas from rounding semidefinite programs in the approximation algorithms literature. Suppose that we have a set of triplet constraints with margin $\alpha$, $\set{(i_t,j_t^+,k_t^-)_{\alpha}}_{t=1}^{m}$ and an embedding $\phi:[n]\to \mathbb{R}^{n}$ that satisfies all the constraints. Then, as we show, one can \emph{round} the embedding $\phi$ to a \emph{one-dimensional embedding} $\hat{\phi}:[n]\to\mathbb{R}$ so that a fraction of at least $\frac{1}{2}+\Omega({\alpha})$ of the triplet constraints without margin $\set{(i_t,j_t^{+},k_t^{-})}_{t=1}^{m}$ are satisfied by $\hat{\phi}$. 
A priori this result seems weak: not only the constraints that $\hat{\phi}$ satisfies are without margin, but also only a fraction of $\frac{1}{2}+\Omega(\alpha)$ of them are satisfied.\footnote{For comparison, a trivial one-dimensional embedding that maps every item to a uniform number in $[0,1]$ satisfies half of the constraints in expectation.}
The fact, however, that there always exists such an embedding in one dimension that is strictly better than the random solution ends up being strong enough to give us the result. Our rounding algorithm projects all the embedded vectors along a random direction in a way similar to the MAX-CUT rounding algorithm of ~\cite{goemans1995improved}, see Lemma \ref{lem:better than random} for the formal statement and proof.

\paragraph{No 1D embedding can fit noise on $n/\varepsilon^2$ triplets.} 
We now go back to our setting where we are given a set of $m=\frac{n}{\varepsilon^2}$ unlabeled triplets, adversarially selected, and we want to show that there is a labeling of the triplets such that no embedding, even in $n$ dimensions, can satisfy \textit{all} of them. The key idea is that our rounding algorithm allows us to focus on the real line and forget about the margin parameter.
Recall that ~\cite{alon2024optimal} show that the VC-dimension of embeddings in dimension $d$ is $\Theta(nd)$, and so for $d=1$, we get that the VC dimension is $\approx n$.
What this means in particular, is that no embedding on the real line is able to fit $\approx \frac{n}{\varepsilon^2}$ triplets labeled by random noise: if we select the labels of a set of $m={n}/{\varepsilon^2}$ triplets randomly then  with constant probability, no embedding will be able to satisfy more than $\frac{1}{2}+\varepsilon$ of them. This is formalized by the notion of Rademacher complexity and its bound in terms of the VC-dimension of a hypothesis class ~\citep{shalev2014understanding,bousquet2003introduction}. Using the probabilistic method, we get that for every set of $m=\frac{n}{\varepsilon^2}$ unlabeled triplets, there exists a labeling such that no embedding on the real line satisfies more than $(\frac{1}{2}+\varepsilon)m$ of them. The formal statement and proof can be found in Lemma \ref{lem: Close to random}.

\subsection{Putting things together} 

We now discuss how to combine the ideas from the above paragraphs to get the result. Let $\set{(i_t,j_t,k_t)}_{t=1}^{m}$ be the set of adversarially selected triplets and suppose, for the sake of contradiction, that the set is shattered by the hypothesis class $\mathcal{H}(n,d,\alpha)$. On the one hand, we get that since $m$ is large enough and linear in $n$, there exists a labeling of the triplets, such that no embedding on the real line (no $h\in \mathcal{H}(n,1)$) satisfies more than a fraction of $\frac{1}{2}+\varepsilon$ of them. Consider that labeling for the triplets. By our assumption that the set is shattered there exists an embedding (potentially in $n$ dimensions) such that all the triplets (with margin) are satisfied. Applying our rounding algorithm on that embedding we get that there exists an embedding on the real line satisfying a fraction of $\frac{1}{2}+\Omega(\alpha)$ of them (without margin). This leads to a contradiction that:
\begin{align*}
    \frac{1}{2}+\Omega(\alpha)\leq \frac{1}{2}+\varepsilon,
\end{align*}
yet we can choose $\varepsilon=\Theta(\alpha)$ to be sufficiently small compared to $\alpha$. The above are made formal in our proof of Theorem \ref{thrm: main_theorem} in Section \ref{sec:proof}.

\section{Improved VC-dimension: Upper bound of Theorem~\ref{th:intro}}\label{sec:proof}
\label{sec: proof of main theorem}

 In this section, we prove the following theorem.
 \begin{theorem}[Upper bound for $\mathrm{VCdim}(\mathcal{H}(n,d,\alpha))$]
 \label{thrm: main_theorem}
     The VC-dimension of $\mathcal{H}(n,d,\alpha)$ is linear in $n$ and independent of the dimensionality $d$. In particular, there exists an absolute constant $C$ such that:
     \begin{align*}
         \mathrm{VCdim}(\mathcal{H}(n,d,\alpha))\leq  C\max\set{n,\frac{n}{\alpha^2}}.
     \end{align*}
\end{theorem}

%\section{Proof of Theorem \ref{thrm: main_theorem}}
Note that the VC dimension is monotonically decreasing in $\alpha$, thus it suffices to show that for $\alpha\in (0,1]$ it is bounded by $C\frac{n}{\alpha^2}$. From now on we assume that $\alpha \in (0,1]$.
Our proof uses the following lemma which, roughly speaking, states that if a set of contrastive learning constraints with margin $\alpha$, $\set{(i_t, j_t^{+}, k_t^{-})_{\alpha}}_{t=1}^{m}$ are satisfiable (potentially in $n$ dimensions) then there exists an embedding on the real line satisfying a $\frac{1}{2}+\Omega(\alpha)$ fraction of the contrastive constraints $\set{(i_t,j_t^{+}, k_t^{-})}_{t=1}^{m}$. Note here that the constraints that $\hat{\phi}$ satisfies are without margin.
\begin{lemma}[Rounding to 1D]\label{lem:better than random}
    Let $\phi:[n]\to \mathbb{R}^n$ be such that the constraints $\set{(i_t,j_t^{+},k_t^{-})_{\alpha}}_{t=1}^{m}$ are satisfied by $\phi$. Then there exists a $\hat{\phi}:[n]\to \mathbb{R}$ such that at least  $ \left(\frac{1}{2}+\frac{\alpha}{5}\right)\cdot m$ constraints of $\set{(i_t,j_t^{+},k_t^{-})}_{t=1}^{m}$ are satisfied by $\hat{\phi}$. 
\end{lemma}
\begin{proof}
    Let $g\sim \mathcal{N}(0, I_n)$ be a standard Gaussian vector. We will consider the random embedding $\hat{\phi}_g:[n]\to \mathbb{R}$ given by $\hat{\phi}(i)=\Iprod{\phi(i), g}$ and show that on expectation, the embedding satisfies a fraction of $\frac{1}{2}+\frac{\alpha}{5}$ of constraints in $\set{(i_t,j_t^{+}, k_t^{-})}_{t=1}^{m}$.  Note that by virtue of the fact that $\|\phi(i)-\phi(k)\|\ne \|\phi(i)-\phi(j)\|$ for every $i,j,k$ with $j\ne k$, we get that, the with probability one, $|\hat{\phi}_g(i)-\hat{\phi}_g(j)|\neq |\hat{\phi}_g(i)-\hat{\phi}_g(k)|$. The above combined imply that there exists a $v$ such that the embedding $\hat{\phi}$ given by $\hat{\phi}(i)=\Iprod{\phi(i), v}$ satisfies a fraction of at least $\frac{1}{2}+\frac{\alpha}{5}$ of constraints and furthermore $h_{\hat{\phi}}\in \mathcal{H}(n,1)$. 

    For $i\in [n]$ let $\bar{i}=\phi(i)$. Observe that the embedding $\hat{\phi}$ satisfies a constraint $(i, j^+, k^-)$ if: \[\left|\Iprod{\bar{i}-\bar{k}, g}\right|>\left|\Iprod{\bar{i}-\bar{j}, g}\right|.\]
    By linearity of expectation, the expected number of constraints that $\hat{\phi}_g$ satisfies is \[\sum_{t=1}^{m}\Pr{\left|\Iprod{\bar{i}_t-\bar{k}_t, g}\right|>\left|\Iprod{\bar{i}_t-\bar{j}_t, g}\right|},\] where the probability is over the randomness in $g$. It therefore suffices to show that for every $t\in [m]$:
    \begin{align*}
        \Pr{\left|\Iprod{\bar{i}_t-\bar{k}_t, g}\right|>\left|\Iprod{\bar{i}_t-\bar{j}_t, g}\right|}\geq \frac{1}{2}+\frac{\alpha}{5}, 
    \end{align*}
    where for vectors $u=\bar{i}_t-\bar{k}_t$ and $v=\bar{i}_t-\bar{j}_t$ it holds that $\|u\|>(1+\alpha)\|v\|.$ To that end, we prove the following claim.
    \begin{claim}\label{clm: 1D projection}
        Let $u,v \in \mathbb{R}^d$ and let $g$ be a standard Gaussian vector.
If $\|u\| > (1+\alpha)\|v\|$ for some $\alpha > 0$, then
\[
\Pr{
|\langle u,g\rangle| > |\langle v,g\rangle| }
\;\ge\;
\frac{2}{\pi}\arctan(1+\alpha).
\]
In particular, for $\alpha\in[0,1]$:
\begin{align*}
    \Pr{
|\langle u,g\rangle| > |\langle v,g\rangle| }
\;\ge\;\frac{1}{2}+\frac{\alpha}{5}.
\end{align*}
\end{claim}
Claim \ref{clm: 1D projection}, which we prove in Appendix \ref{App: Reduction to 1D}, gives us the result.\footnote{Observe this claim is for two vectors $u,v$ as long as they have a gap $\alpha$, so it also allows us to capture the case of quadruplet comparisons $(i,j,k,l)$~\citep{bilu2005monotone,vankadara2023insights}, where we compare the distance between $i,j$ to the distance between $k,l$.}
\end{proof}
In the next lemma we use the upper bound on the VC dimension of contrastive learning on the real line, proven in \cite{alon2024optimal}, to prove that for every set of $\approx \frac{n}{\varepsilon^2}$ unlabeled triplets $(i_t,j_t,k_t)$, there exists a labeling of them such that no embedding is consistent with more than a fraction of $\frac{1}{2}+\varepsilon$ of them. 
\begin{lemma}\label{lem: Close to random}
There exists an absolute constant $c$ such that for every $\varepsilon>0$, the following holds. For every $m\geq \frac{c}{\varepsilon^2}n$ and set $\set{(i_t,j_t,k_t)}_{t=1}^{m}$ there exists a labeling $y:[m]\to \set{-1,1}$ such that no embedding $\phi:[n]\to \mathbb{R}$ satisfies more than a fraction of $\frac{1}{2}+\varepsilon$ of the geometric triplet constraints induced by the labeling. In particular, for every $h\in \mathcal{H}(n,1)$ it holds that:
\begin{align*}
    \frac{1}{m}\sum_{t=1}^{m}\mathbf{1}\set{h(i_t,j_t,k_t)=y(t)}\le \frac{1}{2}+\varepsilon
\end{align*}
\end{lemma}
\begin{proof}
    We will use the probabilistic method. In particular, we will pick the label $y$ randomly and show that \begin{align}\label{eq: almost_Rademacher}\Esymb_{y}\left[\sup_{h\in \mathcal{H}(n,1)} \frac{1}{m}\sum_{t=1}^{m}\mathbf{1}\set{h(i_t,j_t,k_t)=y(t)}\right]\leq \frac{1}{2}+\varepsilon,\end{align} using bounds for the worst case empirical Rademacher complexity of a hypothesis class in terms of the VC-dimension, and then the result directly follows. In particular, we let $y(t)=1$ with probability $1/2$ and $y(t)=-1$ with the remaining probability, independently of everything else. We can equivalently write Equation \ref{eq: almost_Rademacher} as:
    \begin{align}\label{eq: Rademacher on the line}
         \Esymb_{y}\left[\sup_{h\in \mathcal{H}(n,1)}\frac{1}{m}\sum_{t=1}^{m}h(i_t,j_t,k_t)\cdot y(t)\right]\le 2{\varepsilon},
    \end{align}
    and observe that the left hand side exactly corresponds to the empirical Rademacher complexity of hypothesis class $\mathcal{H}(n,1)$ on points $\set{(i_t,j_t,k_t)}_{t=1}^{m}$. We now use (see e.g. \cite{bousquet2003introduction}) that there exists an absolute constant $c'$ such that the empirical Rademacher complexity of a hypothesis class $\mathcal{H}$ on $m$ points can be upper bounded in terms of the VC dimension of the class. In particular, if $\mathcal{R}_m$ is the worst-case\footnote{Supremum over sets of size $m$.} empirical Rademacher complexity on $m$ points then  \begin{align}\label{eq: worst-case Rademacher complexity bound} \mathcal{R}_m\leq c'\sqrt{\frac{\textrm{VCdim}(\mathcal{H})}{m}}.
    \end{align}
    We apply this bound to the left hand side of  Equation \ref{eq: Rademacher on the line} and use that $\VC(\mathcal{H}(n,1))=O(n)$ to get that for $c$ being a sufficiently large constant:
    \begin{align*}
         \Esymb_{y}\left[\sup_{h\in \mathcal{H}(n,1)}\frac{1}{m}\sum_{t=1}^{m}h(i_t,j_t,k_t)\cdot y(t)\right]&\le c'\sqrt{\frac{c_0n\varepsilon^2}{cn}} \leq 2{\varepsilon},
    \end{align*}
    which gives us the result.
\end{proof}
We can now prove our main Theorem \ref{thrm: main_theorem}.
\begin{proof}[Proof of Theorem \ref{thrm: main_theorem}]
Let $\varepsilon<\frac{\alpha}{5}$ and let $c$ be a large enough constant. We will show that no set of size $m\geq c\frac{n}{\varepsilon^2}=\Omega(\frac{n}{\alpha^2})$ can be shattered, which gives the result. 
Consider a set of triplets $\set{(i_t,j_t,k_t)}_{t=1}^{m}$ and assume, for the sake of contradiction, that the set is shattered by $\mathcal{H}(n,d,\alpha)$, meaning that for every labeling of the triplets, there is an embedding $\phi:[n]\to \mathbb{R}^{n}$ that is consistent with all the corresponding triplet constraints.
%We will use Lemma \ref{lem:better than random} and Lemma \ref{lem: Close to random} to get a contradiction: On the one hand,  we use Lemma \ref{lem:better than random} to argue that for every labeling there is an embedding to the real line that is strictly better than random, in that it satisfies a fraction of $\frac{1}{2}+\Omega(a)$ constraints.
%On the other hand, by Lemma \ref{lem: Close to random} we have that there is a labeling of the constraints such that every embedding satisfies very close to half of the constraints, i.e., at most a $\frac{1}{2}+\varepsilon$ fraction of them. Putting the two together and using that $\varepsilon<\frac{\alpha}{5}$, we get a contradiction. 

Formally, by Lemma \ref{lem: Close to random}, taking $c$ bigger than the constant in the lemma we get that there exists a labeling $y^*$ such that for every $h\in \mathcal{H}(n,1)$:
\begin{align}\label{eq: final thrm no better than random}
    \frac{1}{m}\sum_{t=1}^{m}\mathbf{1}\set{h(i_t,j_t,k_t)=y^*_t}\leq \frac{1}{2}+\varepsilon.
\end{align}
On the other hand, by our assumption that the set $\set{(i_t,j_t,k_t)}_{t=1}^{m}$ is shattered, we get that there exists an embedding $\phi$ satisfying all the corresponding constraints: if $y^*_t=1$ then  $\|\phi(i_t)-\phi(k_t)\|>(1+\alpha)\|\phi(i_t)-\phi(j_t)\|$, and if $y^*_t=-1$ then $\|\phi(i_t)-\phi(j_t)\|>(1+\alpha)\|\phi(i_t)-\phi(k_t)\|$. 
Without loss of generality and for notational convenience, we can assume that the labeled triplets are $\set{(i_t,j_t^{+},k_t^{-})_{\alpha}}_{t=1}^{m}$; by Lemma \ref{lem:better than random}, there exists a one dimensional embedding $\hat{\phi}:[n]\to \mathbb{R}$ such that a fraction of at least $\frac{1}{2}+\frac{\alpha}{5}$ constraints in $\set{(i_t,j_t^{+},k_t^{-})}$ are satisfied.\footnote{These are constraints without margin.} For the function $h_{\hat{\phi}}\in \mathcal{H}(n,1)$ this implies that:
\begin{align}\label{eq: final theorem strictly better than random}
    \frac{1}{m}\sum_{t=1}^{m}\mathbf{1}\set{h_{\hat{\phi}}(i_t,j_t,k_t)=y^*_t}\geq \frac{1}{2}+\frac{\alpha}{5}.
\end{align}
Combining Equations \ref{eq: final thrm no better than random} and \ref{eq: final theorem strictly better than random} we get that:
\begin{align*}
&\frac{1}{2}+\frac{\alpha}{5}\leq \frac{1}{2}+\varepsilon,
\end{align*}
which is a contradiction since we have taken $\epsilon<\frac{\alpha}{5}.$
\end{proof}
\section{A Matching Lower Bound}
\label{sec: app lower bound proof}
In this section we prove a matching lower bound for the VC dimension of contrastive learning with margin parameter $\alpha$.
In particular we show that  as long as $\alpha\geq \Omega(\max (\frac{1}{\sqrt{n}}, \frac{1}{\sqrt{d}})),$ then $\mathrm{VCdim}(\mathcal{H}(n,d,\alpha))\geq \Omega(\max\{\frac{n}{\alpha^2}, n\})$.
\begin{theorem}
    Suppose that $\alpha\geq \max (n^{-1/2}, d^{-1/2})$. Then, there exists an absolute constant $c$ such that:
    \begin{align*}
        \mathrm{VCdim}(\mathcal{H}(n,d,\alpha))\geq c\frac{n}{\alpha^2}.
    \end{align*}
\end{theorem}
\begin{proof}
        We divide into two cases. We first consider the case where $\alpha\geq\frac{1}{4}$, in which case we show that the VC dimension is $\Omega(n)$. We make a distinction between the elements in $[n]$. We use $i,j$ to denote the first two elements and $\set{k_t}_{t=1}^{n-2}$ for the rest.
    We consider the set of $\Theta(n)$ triplets given by $\set{(i,j,k_t)}_{t=1}^{n-2}$ and show that they are shattered. For a labeling of the triplets, let $A\subseteq [n-2]$ be the $t\in [n]$ such that $(i,j,k_t)$ is labeled as $(i,j^+, k_t^-)$ and $B=[n-2]\setminus A$.
    We define the embedding to be $\phi(i)=0$, $\phi(j)=1$, furthermore for $t\in A$ we let $\phi(k_t)=\kappa$ where $\kappa>1+\alpha$ and for $t\in B$, $\phi(k_t)=\tau$ where $\tau<\frac{1}{1+\alpha}$.
    Consider $t\in A$ and the triplet $(i,j^+,k_t^-)_{\alpha}$, we have that $|\phi(i)-\phi(k_t)|=\kappa>(1+\alpha)|\phi(i)-\phi(j)|$, and the triplet $(i,j^+,k_t^-)_{\alpha}$ is satisfied. On the other hand for $t\in B$, we have that $|\phi(i)-\phi(j)|=\frac{1}{\tau}|\phi(i)-\phi(k_t)|>(1+\alpha)|\phi(i)-\phi(k_t)|$ and the triplet $(i,k_t^+,j^-)_{\alpha}$ is satisfied.
    In the above construction there exist $i,j,k$ with $j\ne k$ such that $\|\phi(i)-\phi(j)\|= \|\phi(i)-\phi(k)\|$. A small perturbation of all points resolves this issue while maintaining the conclusion.
    
    In the remaining case, we consider $\alpha<\frac{1}{4}$. We proceed to describe an instance of $\approx \frac{n}{\alpha^2}$ triplets that is shattered. Let for convenience $\gamma =4\alpha$. We first make a distinction between the elements in $[n]$. We will use $k_1,\ldots,k_{\frac{1}{\gamma^2}}$ to denote the first $\frac{1}{\gamma^2}$ elements of $[n]$ and $i_t,j_t$ with $t\in \left[\frac{n-\frac{1}{\gamma^2}}{2}\right]$ for the remaining elements. The set of triplets that we propose is inspired by the construction of \cite{alon2022theory} for the lower bound on the VC dimension of linear classifiers with margin.
    
    We consider the set of triplets $\set{(i_t,j_t,k_l): t\in \left[\frac{n-\frac{1}{\gamma^2}}{2}\right], l\in [\frac{1}{\gamma^2}]}$ and show that the set is shattered.
    Note that the number of triplets is $\frac{1}{2\gamma^2}(n-\frac{1}{\gamma^2})=\Theta(\frac{n}{\alpha^2})$ (here, we have used that $\frac{1}{\alpha^2}\leq n$), thus showing that they can be shattered gives us the result.
    For every labeling of the triplets, the embedding for points $k_l$ will be the same.
    In particular, we will always map $k_l$ to $e_l,$ the $l$-th standard basis vector (here we have used that $\frac{1}{\alpha^2}\leq d$).  Now consider a labeling of the triplets and observe that we can decompose the triplets as $\cup_{t}\set{(i_t,j_t,k_l)}$, where having fixed the mappings of $k_l$, the instances are decoupled for different $t$.
    It therefore suffices to show that for every labeling of triplets $\set{(i,j,k_l)}$ there is an embedding of points $i$ and $j$ that satisfies all the triplets. 
    
        Consider a labeling of the triplets $\set{(i,j,k_l)}_{l=1}^{\frac{1}{\gamma^2}}$ and let $A$ be the subset of $l$ in $[\frac{1}{\gamma^2}]$ such that the triplet $(i,j,k_l)$ is labeled as $(i,k_l^+,j^-)_{\alpha}$ and $B=[\frac{1}{\gamma^2}]\setminus A$.
    Let $w=\sum_{l\in A}e_l-\sum_{l\in B}e_l$, $w$ is exactly the vector that defines the hyperplane separating the points $\set{e_l, l\in A}$ from the points $\set{e_l, l\in B}$ with margin $\gamma$, see Proposition 17 in \cite{alon2022theory}.
    We let $\phi(i) =\gamma\cdot w$. We also select $j$ to be a point at distance $\sqrt2$ from $\phi(i)$. We now consider $l\in A$, we have that 
    \begin{align*}
        \|\phi(i)-\phi(k_l)\|&=\|\phi(i)-e_l\|\\&
        =\left((1-\gamma)^2+\gamma^2\left(\frac{1}{\gamma^2}-1\right)\right)^{\frac{1}{2}}\\
    &=\sqrt{2-2\gamma}.
    \end{align*}
    This in turn means that:
    \begin{align*}
        \|\phi(i)-\phi(j)\|&=\sqrt{2}\\
        &= \frac{1}{\sqrt{1-\gamma}}\cdot \|\phi(i)-e_l\|\\
        &\geq(1+\frac{1}{2}\gamma)\cdot \|\phi(i)-e_l\|\\
        &=(1+2\alpha)\cdot \|\phi(i)-\phi(k_l)\|
    \end{align*}
    where we have used that the function $f(x)=\frac{1}{\sqrt{1-x}}$ is convex and $f'(0)=\frac{1}{2}$. We therefore get that $(i,k_l^+, j^-)_{\alpha}$ is satisfied by the embedding.
    
    On the other hand, for $l\in B$, we have that 
    \begin{align*}
        \|\phi(i)-\phi(k_l)\|&=\sqrt{1+\gamma}\cdot \sqrt{2}\\
        &=\sqrt{1+\gamma}\cdot \|\phi(i)-\phi(j)\|\\
        &\geq (1+\frac{\gamma}{4})\cdot \|\phi(i)-\phi(j)\|\\
        &=(1+\alpha)\cdot \|\phi(i)-\phi(j)\|,
    \end{align*}
    meaning that the triplet $(i,j^+,k_l^-)_{\alpha}$ is satisfied. In the above construction there exist $i,j,k$ with $j\ne k$ such that $\|\phi(i)-\phi(j)\|= \|\phi(i)-\phi(k)\|$. A small perturbation of all points resolves this issue while maintaining the conclusion.

\end{proof}
\section{Experiments}\label{sec:experiments}
%say broader comment on practice of ML that meaningful embeddings are separated by at least some amount, so we expalin, also mention we tested experimentally

To complement our theoretical bounds, we empirically explore how the generalization error $\hat \epsilon$ is affected as we vary the dimension $d$, the number of points $n$, the margin $\alpha$ and the number $m$ of given triplet samples. We test both on synthetic data generated from a Gaussian distribution and on real datasets (CIFAR-100). Our goal is to test the dimension-independence of our $O(n/\alpha^2)$ bound.

\paragraph{Experimental Setting.} For each configuration $(n,d,\alpha)$, we sample $m$ training triplets from $[n]^3$ uniformly at random, reject any that fail the $(1+\alpha)$-margin condition under a ground-truth metric $\phi^*$, and fit an embedding $\hat\phi:[n]\to\R^{d}$ by minimizing the multiplicative-margin triplet loss
\[
\mathcal{L}(\hat\phi) \;=\; \frac{1}{m}\sum_{t=1}^{m}\Brac{(1+\alpha)^2\bigsnormt{\hat\phi(i_t)-\hat\phi(j_t)} - \bigsnormt{\hat\phi(i_t)-\hat\phi(k_t)}}_+,
\]
which is the natural surrogate for the $(1+\alpha)$-separation defining $\mathcal{H}(n,d,\alpha)$ (see~\cite{schroff2015facenet} and survey~\cite{vankadara2023insights}). We optimize with Adam, with learning rate $0.01$ and $1000$ steps, and report the empirical generalization error $\hat\eps$ on $5{,}000$ held-out $(1+\alpha)$-separated test triplets. 

We perform experiments for two types of ground-truth embeddings. In the \emph{synthetic} setting, $\phi^*$ is the undelying embedding of $n$ points drawn i.i.d. from $\calN(0,I_{50})$ and the learner directly parameterizes the $n$ embedding rows in $\R^{d}$. In the \emph{real-data} setting, $\phi^*$ is the L2-normalized penultimate feature (this provides a vector in $\mathbb{R}^{512}$) of an ImageNet-pretrained ResNet18~\citep{he2016deep} on $5{,}000$ CIFAR-100 test images~\citep{krizhevsky2009learning}, frozen during training, and the learner is a linear head $\R^{512}\to\R^{d}$ on top of these features.

\begin{figure}[ht]
\centering
\begin{minipage}{0.88\textwidth}
\centering
\begin{subfigure}{0.47\linewidth}
  \centering
  \includegraphics[width=\linewidth]{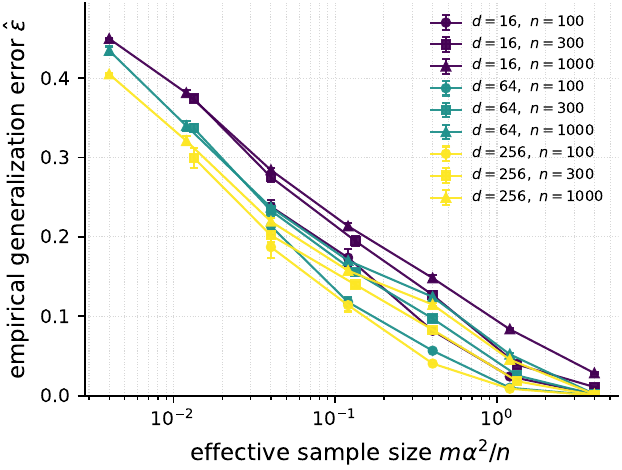}
  \caption{Synthetic, $\alpha=0.2$.}
  \label{fig:synth-dim}
\end{subfigure}\hfill
\begin{subfigure}{0.49\linewidth}
  \centering
  \includegraphics[width=\linewidth]{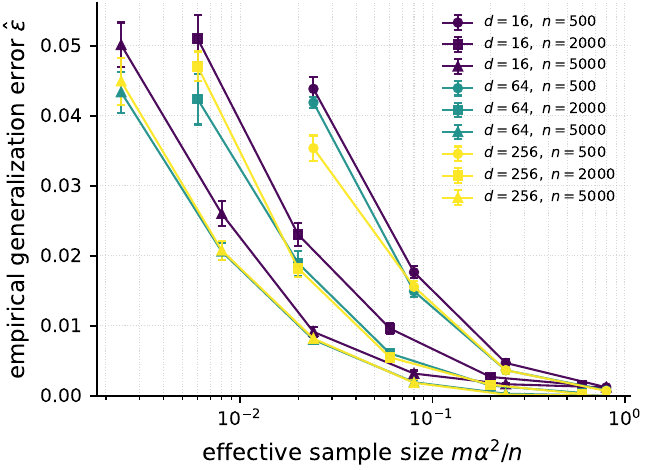}
  \caption{CIFAR-100 / ResNet18, $\alpha=0.2$.}
  \label{fig:real-dim}
\end{subfigure}
\end{minipage}

\caption{Dimension-independence: empirical generalization error
$\hat\eps$ against $m\alpha^2/n$ at $\alpha=0.2$.
\textbf{(a)}~Synthetic setup with underlying embedding
$\phi^\ast\in\R^{50}$.
\textbf{(b)}~CIFAR-100 / ResNet18 setup with a linear head over
512-dimensional features.}
\label{fig:dim-pair}
\end{figure}

\paragraph{Dimension-independence.} As Theorem~\ref{thrm: main_theorem} does not depend on the embedding dimension $d$, for fixed $m\alpha^2/n$, the empirical error should not change as we increase $d$. To test this we fix $\alpha=0.2$ and vary $d\in\{16,64,256\}$. On the synthetic side we keep the ground-truth ambient dimension at $50$ so that the data distribution itself is unchanged, and only the learner's capacity changes. We sweep $n\in\{100,300,1000\}$ and $m\in\{100,300,10^3,3\!\cdot\!10^3,10^4,3\!\cdot\!10^4,10^5\}$ on the synthetic setup, and $n\in\{500,2000,5000\}$ and $m\in\{300,10^3,3\!\cdot\!10^3,10^4,3\!\cdot\!10^4,10^5\}$ on the real-data setup, over 5 random seeds. Figure~\ref{fig:dim-pair} agrees with the theoretical predictions on both setups, as it shows the three curves corresponding to different dimensions $d=16,64,256$ (encoded by different colors) lie on top of each other within each $n$-trajectory. Interestingly, the plots show that on real-data (Fig.~1 (b)), dimension-independence is quite prominent, since as $d$ varies, the error $\hat \epsilon$ changes only slightly (notice the small values of $\hat \epsilon$ on the $y$-axis). Appendix~\ref{app:additional-experiments} provides further experiments.%: we provide $\alpha$-sweeps, varying the size of the margin, where we verify the $m\alpha^2/n$ scaling, additional $\alpha$-slices that show dimension-vs-generalization error, and a margin-satisfaction sanity check for learned embeddings.
\section*{Conclusion}
We resolved an open question from~\cite{alon2024optimal} related to provable generalization guarantees and the VC-dimension of contrastive learning, in the presence of a margin parameter separating the distance between an anchor-positive pair from an anchor-negative pair. Enforcing a margin is not only a practically well-motivated approach in the literature of contrastive learning and ordinal embeddings, but as we show it also leads to interesting theoretical insights. Our main result improves upon the direct application of standard low-distortion metric embedding techniques, such as the Johnson-Lindenstrauss lemma, and provides an optimal $\Theta(\max\set{n/\alpha^2,n})$ bound for the VC-dimension for learning a ground-truth embedding on $n$ data points of interest (with the lower bound holding under the assumption that $\alpha\geq \max\set{n^{-1/2},d^{-1/2}}$). Surprisingly, this is independent of the underlying dimensionality $d$ of the ground-truth embedding, which is often large. Given that our bound is optimal in $n$,  our work makes steps towards understanding the empirical success of contrastive learning. Technically, our approach uses ideas from approximation algorithms and geometric projections, coupled with learning theory bounds for learning triplet  comparisons.

\section*{AI Disclosure}
Generative AI tools were used for polishing and for feedback on exposition as well as literature search. In particular, the free version of ChatGPT helped the authors find a reference for the stronger bound on Rademacher complexity in terms of VC-dimension. The main result of the paper was obtained by the authors in April 2026. %without the use of AI
 All the proofs and technical ideas come from the authors. Generative AI tools were used to assist with the experiments. 
\newpage

\newpage

\bibliographystyle{abbrvnat}
\bibliography{references.bib}
\newpage
\appendix
\begin{appendices}
\section{Reduction to 1D}
\label{App: Reduction to 1D}
\iffalse\begin{lemma}
Let $u,v \in \mathbb{R}^d$ and let $g$ be a standard Gaussian vector.
If $\|u\| \ge (1+\alpha)\|v\|$ for some $\alpha > 0$, then
\[
\Pr{
|\langle u,g\rangle| \ge |\langle v,g\rangle| }
\;\ge\;
\frac{2}{\pi}\arctan(1+\alpha).
\]
In particular, for $\alpha\in[0,1]$:
\begin{align*}
    \Pr{
|\langle u,g\rangle| \ge |\langle v,g\rangle| }
\;\ge\;\frac{1}{2}+\frac{2}{\pi}\left(\arctan(2)-\frac{\pi}{4}\right)\alpha
\end{align*}
where $\frac{2}{\pi}\left(\arctan(2)-\frac{\pi}{4}\right)\approx0.204833\geq \frac{1}{5}$.
\end{lemma}
\fi

In this section we prove Claim \ref{clm: 1D projection}.
\begin{proof}[Proof of Claim \ref{clm: 1D projection}]
If $v=0$, the claim is immediate. Assume $v \neq 0$, which also implies $u \neq 0$. Let $\bar u = u/\|u\|$ and $\bar v = v/\|v\|$. Define $\beta = 1+\alpha$. Then the inequality $|\langle u,g\rangle| > |\langle v,g\rangle|$ holds if and only if 
$$
|\langle \bar u,g\rangle| > \frac{\|v\|}{\|u\|} \cdot |\langle \bar v,g\rangle|.
$$
Since $\|u\| > \beta\|v\|$,
it suffices to show
\[
\Pr{|\langle \bar v,g\rangle| < \beta|\langle \bar u,g\rangle| }
\;\ge\;
\frac{2}{\pi}\arctan(\beta).
\]

If $\bar v = \pm \bar u$, then the event holds with probability $1$, so we may assume that $\bar u$ and $\bar v$ are not collinear. Let $S$ be the two-dimensional subspace spanned by $\bar u$ and $\bar v$. We decompose $g$ as 
$$
g = \tilde g + g^\perp,
$$
where $\tilde g$ is the orthogonal projection of $g$ onto $S$. Since both $\bar u$ and $\bar v$ lie in $S$, we have
\[
\langle \bar u,g\rangle = \langle \bar u,\tilde g\rangle,
\quad
\langle \bar v,g\rangle = \langle \bar v,\tilde g\rangle.
\]

Let $\theta \in (0,\pi)$ be the angle between $\bar u$ and $\bar v$, measured counterclockwise, and let $\varphi$ be the angle from $\bar u$ to $\tilde g$, also measured counterclockwise. To visualize the configuration, align the first basis vector $e_1$ of $S$ with $\bar u$. Then $\bar u$ has coordinates $(1,0)$, $\bar v$ has coordinates $(\cos \theta, \sin \theta)$, and $\tilde g$ has coordinates $(\cos \varphi, \sin \varphi)$. See Figure~\ref{fig:geom-and-F}. We exclude the events $\tilde g = 0$ and $\varphi \in \{\pi/2, 3\pi/2\}$, which occur with probability $0$. Then $\varphi$ is uniformly distributed on
\[
A = (-\pi/2,\pi/2)\cup (\pi/2,3\pi/2).
\]

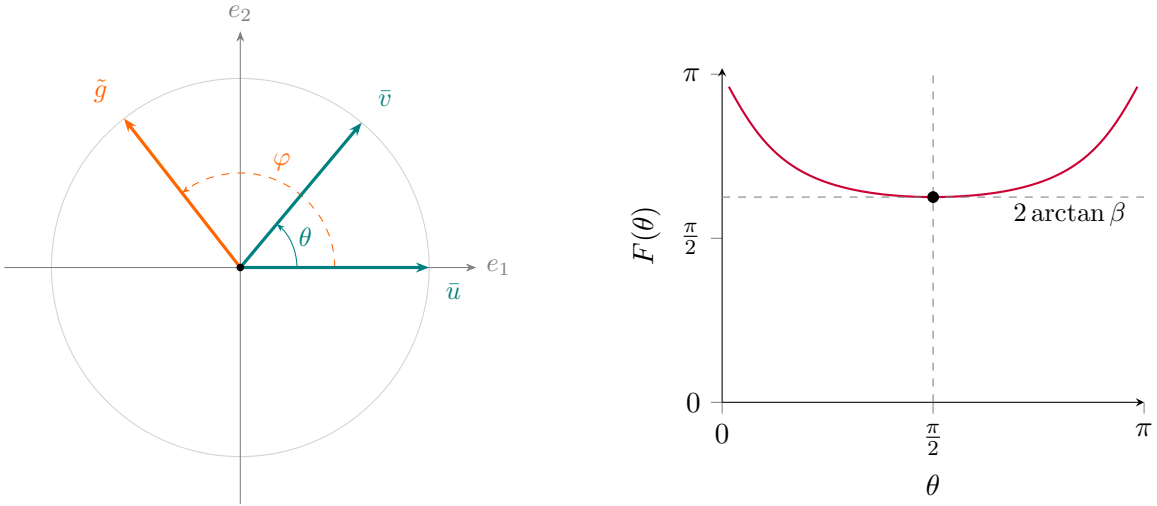
\begin{figure}[ht]
\centering
\begin{subfigure}[t]{0.5\textwidth}
    \centering
    \begin{tikzpicture}[scale=2.5,
      vec/.style={-{Stealth[length=5pt]}, line width=1.2pt},
      axis/.style={-{Stealth[length=4pt]}, thin, gray},
      lbl/.style={font=\small}
    ]
      \def\thetaAngle{50}
      \def\phiAngle{128}

      % --- Axes ---
      \draw[axis] (-1.25,0) -- (1.25,0) node[right, lbl] {$e_1$};
      \draw[axis] (0,-1.25) -- (0,1.25) node[above, lbl] {$e_2$};

      % --- Unit circle ---
      \draw[gray, thin, opacity=0.35] (0,0) circle (1);

      % --- Arc for theta ---
      \draw[-{Stealth[length=3pt]}, teal, thin]
        (0.30,0) arc[start angle=0, end angle=\thetaAngle, radius=0.30];

      % --- Improved theta label (moved outward) ---
      \node[lbl, teal]
        at ({0.38*cos(\thetaAngle/2)}, {0.38*sin(\thetaAngle/2)})
        {$\theta$};

      % --- Arc for phi ---
      \draw[-{Stealth[length=3pt]}, orange!80!red, thin, dashed]
        (0.50,0) arc[start angle=0, end angle=\phiAngle, radius=0.50];

      \node[lbl, orange!80!red]
        at ({0.5*cos(\phiAngle/2)}, {0.62*sin(\phiAngle/2)})
        {$\varphi$};

      % --- Vector u-bar ---
      \draw[vec, teal] (0,0) -- (1,0)
        node[below right=2pt, lbl, teal] {$\bar{u}$};

      % --- Vector v-bar ---
      \draw[vec, teal] (0,0) -- (\thetaAngle:1)
        node[above right=2pt, lbl, teal] {$\bar{v}$};

      % --- Vector g-tilde ---
      \draw[vec, orange!80!red] (0,0) -- (\phiAngle:1)
        node[above left=2pt, lbl, orange!80!red] {$\tilde{g}$};

      % --- Origin dot ---
      \fill (0,0) circle (0.02);
    \end{tikzpicture}
\end{subfigure}%
\hspace{0.02\textwidth}
\begin{subfigure}[t]{0.45\textwidth}
    \centering
    \def\betaplot{1.5}
    \begin{tikzpicture}
    \begin{axis}[
        width=0.96\linewidth,
        height=6cm,
        axis lines=left,
        xlabel={$\theta$},
        ylabel={$F(\theta)$}, % <-- fixed here
        xmin=0, xmax=3.14159,
        ymin=0, ymax=3.2,
        domain=0.05:3.09159,
        samples=250,
        xtick={0, 1.5707963, 3.14159},
        xticklabels={$0$, $\tfrac{\pi}{2}$, $\pi$},
        ytick={0, 1.5707963, 3.14159},
        yticklabels={$0$, $\tfrac{\pi}{2}$, $\pi$},
        tick align=outside,
        enlargelimits=false
    ]
      % --- F(theta) ---
      \addplot[thick, purple!80!red]
        {rad(atan((\betaplot - cos(deg(x)))/sin(deg(x)))
           - atan((-\betaplot - cos(deg(x)))/sin(deg(x))))};

      % --- Vertical reference line ---
      \addplot[dashed, gray, thin]
        coordinates {(1.5707963, 0) (1.5707963, 3.2)};

      % --- Horizontal reference line ---
      \addplot[dashed, gray, thin]
        coordinates {(0, {2*rad(atan(\betaplot))}) (3.14159, {2*rad(atan(\betaplot))})};

      % --- Minimum point ---
      \addplot[only marks, mark=*, mark size=2pt, black]
        coordinates {(1.5707963, {2*rad(atan(\betaplot))})};

      % --- Label inside plot (no overflow) ---
      \node[font=\small, anchor=north east, inner sep=2pt]
        at (axis cs:3.05, {2*rad(atan(\betaplot))})
        {$2\arctan\beta$};
    \end{axis}
    \end{tikzpicture}
\end{subfigure}

\caption{Left: vectors $\bar{u}$, $\bar{v}$, and $\tilde{g}$ in the plane, with angles
$\theta$ and $\varphi$. Right: the function $F(\theta)$ for $\beta = 1.5$,
illustrating that the minimum is attained at $\theta = \pi/2$.}
\label{fig:geom-and-F}
\end{figure}

We have
\[
\langle \tilde g,\bar u\rangle = \|\tilde{g}\| \cos\varphi,
\quad
\langle \tilde g,\bar v\rangle = \|\tilde{g}\|\cos(\varphi - \theta).
\]
Therefore,
\[
|\langle \bar v,g\rangle| < \beta|\langle \bar u,g\rangle|
\;\Longleftrightarrow\;
|\cos(\varphi-\theta)| < \beta|\cos\varphi|.
\]

Define
\[
\mathcal{E} = \left\{ \varphi \in A : |\cos(\varphi-\theta)| < \beta|\cos\varphi| \right\}.
\]
Since $\cos\varphi \neq 0$, this is equivalent to
\[
\mathcal{E} = \Big\{ \varphi \in A : 
-\beta < 
\frac{\cos(\varphi-\theta)}{\cos \varphi} < \beta \Big\}.
\]
Using $\cos(\varphi-\theta) = \cos \varphi\cos \theta + \sin \varphi \sin \theta$, we get
\[
\mathcal{E} = \Big\{ \varphi \in A : 
-\beta <
\cos \theta + \tan(\varphi)\sin \theta < \beta \Big\}.
\]
Since $\theta \in (0,\pi)$, we have $\sin\theta > 0$, and hence
\[
\mathcal{E} = \Big\{ \varphi \in A : 
\frac{- \beta - \cos \theta}{\sin \theta} < 
 \tan(\varphi)
<\frac{\beta - \cos \theta}{\sin \theta} \Big\}.
\]

Therefore, $\mathcal E$ consists of two intervals $I$ and $I +\pi$, where
\[
I = 
\Bigl(\;\arctan\Big(\frac{
-\beta-\cos \theta}
{\sin \theta}\Big),
\arctan\Big(\frac{
\beta-\cos \theta}
{\sin \theta}\Big)\;\Bigr).
\]
The probability of $\varphi\in \mathcal{E}$ equals 
\[
\frac{|\mathcal{E}|}{2\pi}
=
\frac{|I|}{\pi}
=
\frac{1}{\pi}\Big[\arctan\Big(\frac{
\beta-\cos \theta}
{\sin \theta}\Big) - \arctan\Big(\frac{
-\beta-\cos \theta}
{\sin \theta}\Big)\Big].
\]
Denote the expression in the square brackets on the right hand side by $F(\theta)$ (see Figure \ref{fig:geom-and-F}):
\[
F(\theta)
=
\arctan\Big(\frac{\beta-\cos \theta}{\sin \theta}\Big)
-
\arctan\Big(\frac{-\beta-\cos \theta}{\sin \theta}\Big).
\]
Then the probability above is $F(\theta)/\pi$. The derivative of $F(\theta)$ equals
\[
F'(\theta)
=
\frac{1-\beta\cos\theta}{1-2\beta\cos\theta+\beta^2}
-
\frac{1+\beta\cos\theta}{1+2\beta\cos\theta+\beta^2}.
\]
Write it as
\[
F'(\theta) = G(-\beta \cos\theta)- G(\beta \cos\theta),\quad\text{where }
G(x) = \frac{x+1}{2x+1 + \beta^2}. 
\]
Function $G(x)$ is increasing for $x > -(\beta^2 +1)/2$ because 
$$G(x)= \frac{1}{2}\Bigl[1 - \frac{\beta^2-1}{2x + 1 + \beta^2}\Bigr]$$
and $\beta^2 - 1 > 0$. Observe that
$\beta \cos \theta > -(\beta^2 +1)/2$ and $-\beta \cos \theta > -(\beta^2 +1)/2$. Hence $F'(\theta) < 0$ when $\cos\theta > 0$ and $F'(\theta) > 0$ when $\cos\theta < 0$. Therefore, $F$ is minimized at $\theta = \pi/2$. Consequently,
\[
F(\theta) \ge F(\pi/2)
=
\arctan(\beta)-\arctan(-\beta)
=
2\arctan(\beta).
\]
Thus
\[
\Pr{\varphi\in\mathcal E}
=
\frac{F(\theta)}{\pi}
\ge
\frac{2}{\pi}\arctan(\beta).
\]
This proves the first part of the claim. We now use that the function $\arctan(x)$ is concave for $x\geq 0$ meaning that in any interval, it can be lower bounded by the secant line between the endpoints of the interval. In particular, we get that for $\alpha\in [0,1]$:
\begin{align*}
    \frac{2}{\pi}\arctan(\beta)&\geq \frac{2}{\pi}\arctan(1+\alpha)\\
    &\geq \frac{2}{\pi}\left((\arctan(2)-\frac{\pi}{4})\alpha+\frac{\pi}{4}\right)\\
    &=\frac{1}{2}+\frac{2}{\pi}\left(\arctan(2)-\frac{\pi}{4}\right)\alpha,
\end{align*}
where $\frac{2}{\pi}\left(\arctan(2)-\frac{\pi}{4}\right)\approx0.204833\geq \frac{1}{5}$.
\end{proof}
\section{Additional Experiments}\label{app:additional-experiments}

In this section we provide the additional experiments referenced from Section~\ref{sec:experiments}. All experiments run on a standard CPU and can be completed within an hour.

\subsection{Synthetic \texorpdfstring{$\alpha$}{alpha}-sweep}\label{app:synth-alpha}

In Section~\ref{sec:experiments}, we fixed $\alpha$ and varied $d$. Here we do the opposite: on the synthetic setup, we fix $d=50$ to match the ground-truth ambient dimension, and sweep $\alpha\in\{0.1,0.2,0.4\}$, $n\in\{100,300,1000\}$, $m\in\{100,300,10^3,{3\!\cdot\!10^3},10^4,{3\!\cdot\!10^4},10^5\}$ over 5 seeds. In Figure~\ref{fig:synthetic}, we plot the empirical error $\hat\eps$ against the effective sample size $m\alpha^2/n$ predicted by Theorem~\ref{thrm: main_theorem}. All nine $(n,\alpha)$ curves collapse onto a single trajectory and drop from random-guess error to near zero once $m\alpha^2/n\gtrsim 1$. If we plot against $m/n$ instead, the curves separate by margin, so the $\alpha^2$ factor in our obtained bound is quite precise.

\begin{figure}[ht]
\centering
\includegraphics[width=0.48\linewidth]{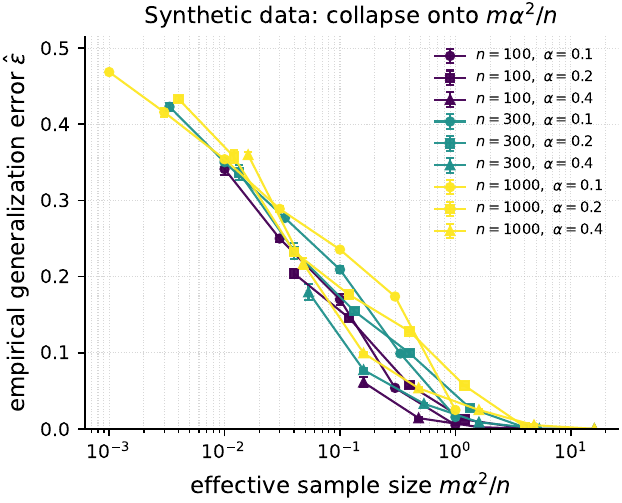}\hfill
\includegraphics[width=0.48\linewidth]{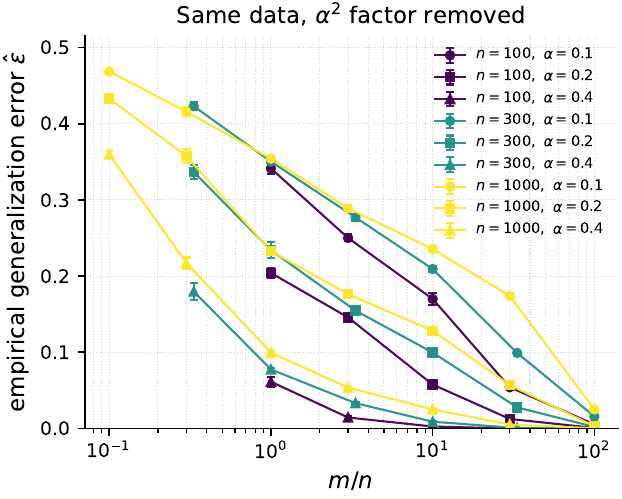}
\caption{Synthetic $\alpha$-sweep. Left: all nine $(n,\alpha)$ curves collapse when plotted against the predicted effective sample size $m\alpha^2/n$. Right: without the $\alpha^2$ factor the curves separate by margin $\alpha$.}
\label{fig:synthetic}
\end{figure}

\subsection{Real data \texorpdfstring{$\alpha$}{alpha}-sweep}\label{app:real-alpha}
On the real-data setup, we fix $d=64$, and sweep $\alpha\in\{0.1,0.2,0.4\}$, $n\in\{500,2000,5000\}$, $m\in\{300,10^3,{3\!\cdot\!10^3},10^4,{3\!\cdot\!10^4},10^5\}$ over 5 seeds. Figure~\ref{fig:real} reproduces both diagnostics: plotting against $m\alpha^2/n$ orders the curves and they decay together to zero error, while removing $\alpha^2$ separates them by $\alpha$. Compared to the synthetic data, two things look different: the empirical errors are well below their synthetic counterparts at the same $m\alpha^2/n$, and at fixed $m\alpha^2/n$ larger $n$ gives \emph{lower} rather than higher error. Both behaviors are what we would expect when the linear head explores a more structured subspace of $\mathcal{H}(n,d,\alpha)$, rather than randomly generated data (it has only $O(d \cdot 512)$ parameters, independent of $n$). A smaller class generalizes better than our worst-case bound predicts, and at fixed $m\alpha^2/n$, a larger value of $n$ means more triplets constraining the same fixed-size parameter matrix.

\begin{figure}[ht]
\centering
\includegraphics[width=0.48\linewidth]{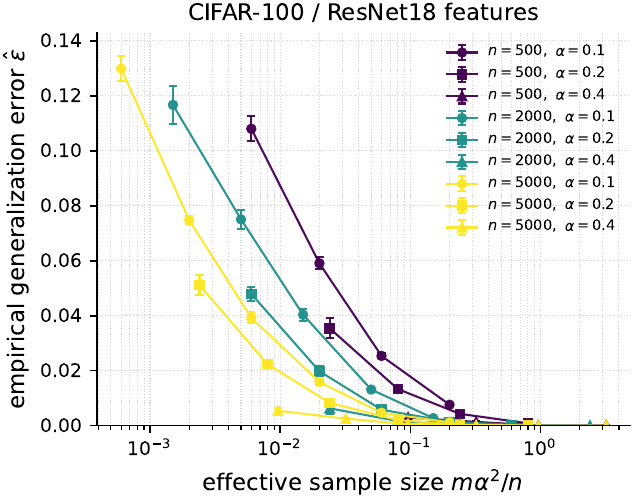}\hfill
\includegraphics[width=0.48\linewidth]{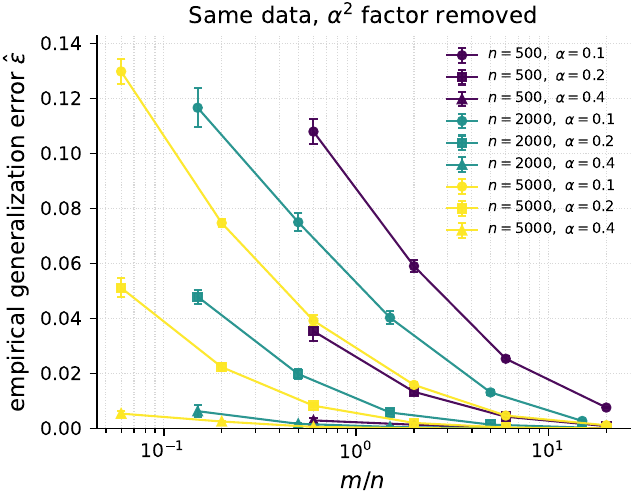}
\caption{Real-data $\alpha$-sweep. Left: collapse against $m\alpha^2/n$. Right: with the $\alpha^2$ factor removed.}
\label{fig:real}
\end{figure}

\subsection{Dimension-independence across \texorpdfstring{$\alpha$}{alpha}}\label{app:dim-multi-alpha}

The dimension-independence experiment in Section~\ref{sec:experiments} uses a single representative margin $\alpha=0.2$. We repeat this experiment at $\alpha\in\{0.05,0.1,0.4\}$ on the synthetic and real-data setups (Figure~\ref{fig:dim-multi-alpha}). The three $d$-curves overlap within each $n$-trajectory across the full range in both cases. The one exception is mild underfitting at $d=16$ when $\alpha$ is small and $m\alpha^2/n$ is large: 16 dimensions cannot faithfully embed the ground truth ($d^\star=50$ synthetic, up to 512 on CIFAR-100), and tight margins are more sensitive to the resulting distortion.

\begin{figure}[ht]
\centering
\includegraphics[width=\textwidth]{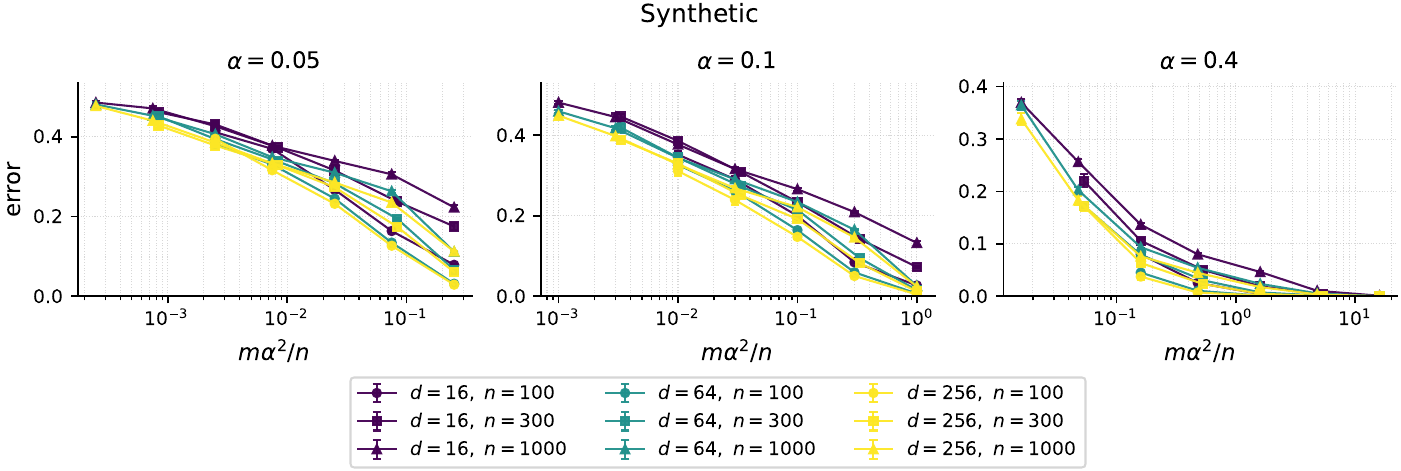}\\[1em]
\includegraphics[width=\textwidth]{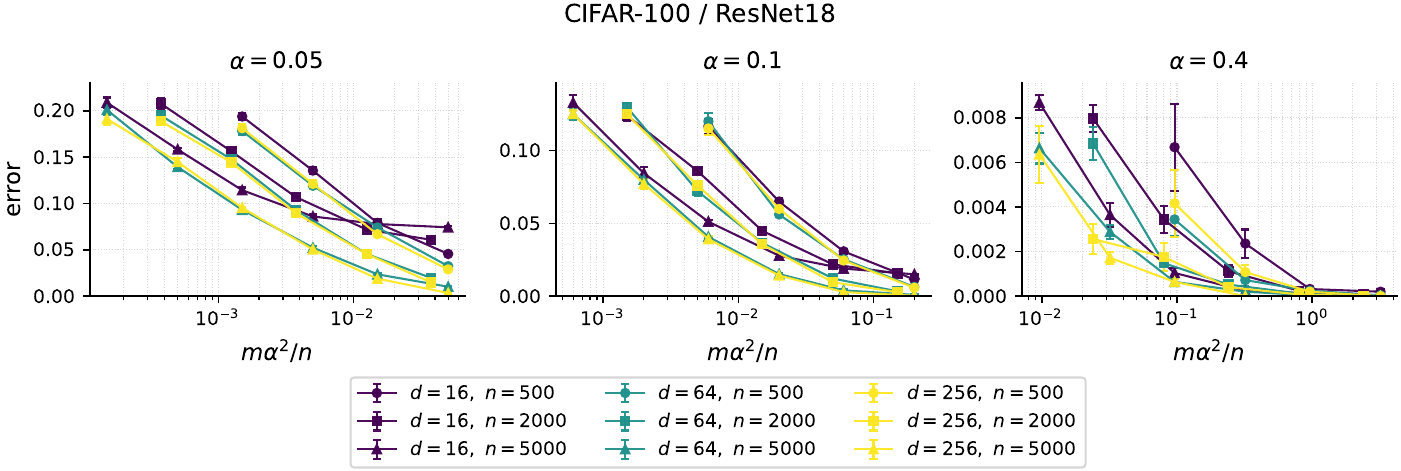}
\caption{Dimension-independence across $\alpha$. Top: synthetic setup. Bottom: CIFAR-100 / ResNet18.}
\label{fig:dim-multi-alpha}
\end{figure}

\newpage

\subsection{Margin sanity check}\label{app:margin-sanity}

The hypothesis class $\mathcal{H}(n,d,\alpha)$ is defined by embeddings whose output itself satisfies the $(1+\alpha)$-separation; the multiplicative-margin loss is a continuous surrogate for this constraint. As a post-hoc check that the learned $\hat\phi$ actually \textit{lives in the hypothesis class}, we measure the fraction of held-out triplets on which the learned embedding satisfies $\norm{\hat\phi(i)-\hat\phi(k^-)}\ge (1+\alpha)\norm{\hat\phi(i)-\hat\phi(j^+)}$ (Figure~\ref{fig:margin-sanity}). The fraction climbs from around $0.5$ at the smallest $m$ to above $0.95$ once $m\alpha^2/n\gtrsim 1$. So whenever the empirical error is small, the learned embedding really is $(1+\alpha)$-separated on most test triplets.

\begin{figure}[ht]
\centering
\includegraphics[width=0.55\textwidth]{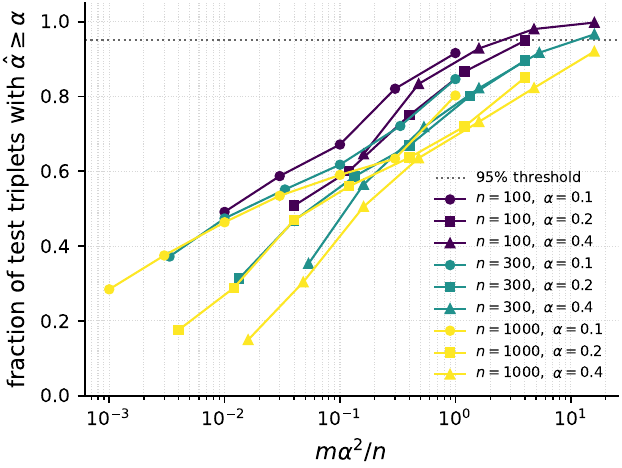}
\caption{Margin sanity check. Once $m\alpha^2/n\gtrsim 1$, the learned embedding satisfies the hypothesis-class constraint on $>\!95\%$ of test triplets.}
\label{fig:margin-sanity}
\end{figure}

\end{appendices}

\end{document}